\documentclass[10pt,conference]{IEEEtran}

\usepackage{mathrsfs}
\usepackage{amsmath}
\usepackage{amssymb}
\usepackage{amsthm}
\usepackage{graphics}
\usepackage{graphicx}
\usepackage{float}
\usepackage{url}
\usepackage{lipsum}
\usepackage{mathrsfs}
\begin{document}

\input epsf

\def\thesection {\arabic{section}}
\newtheorem{definition}{Definition}
\newtheorem{corollary}{Corollary}
\newtheorem{theorem}{Theorem}
\newtheorem{lemma}{Lemma}
\newtheorem{remark}{Remark}
\newtheorem{example}{Example}
\newtheorem{proposition}{Proposition}
\newtheorem{question}{Question}
\newtheorem{conjecture}{Conjecture}

\newcommand{\dsum}{\displaystyle\sum}
\newcommand{\dint}{\displaystyle\int}
\newcommand{\reals}{{\rm I\!R}}
\newcommand{\Prob}{{\rm Prob\,}}
\newcommand{\expectation}{\ensuremath{\mathbb{E}}}
\newcommand{\naturals}{\ensuremath{\mathbb{N}}}
\newcommand{\set}{\ensuremath{\mathcal}}

\title{\huge{Upper Bounds on the Relative Entropy and R\'{e}nyi Divergence as a Function
of Total Variation Distance for Finite Alphabets}}

\author{Igal Sason  \hspace*{6cm} Sergio Verd\'{u}\\
Department of Electrical Engineering \hspace*{2cm} Department of Electrical Engineering\\
\hspace*{-2cm} Technion -- Israel Institute of Technology \hspace*{2.8cm} Princeton University\\
\hspace*{1.5cm} Haifa 32000, Israel \hspace*{3.5cm} Princeton, New Jersey 08544, USA\\
E-mail: \url{sason@ee.technion.ac.il}  \hspace*{2.9cm} E-mail: \url{verdu@princeton.edu}}

\maketitle
\thispagestyle{empty}
\pagestyle{empty}

\begin{abstract}
A new upper bound on the relative entropy is derived as a function of the total
variation distance for probability measures defined on a common finite alphabet.
The bound improves a previously reported bound by Csisz\'{a}r and Talata. It is
further extended to an upper bound on the R\'{e}nyi divergence of an arbitrary
non-negative order (including $\infty$) as a function of the total variation distance.
\end{abstract}

\vspace*{0.1cm}

\thispagestyle{empty}
\pagestyle{empty}

\section{Introduction}
Consider two probability distributions $P$ and $Q$ defined on a common
measurable space $(\set{A}, \mathscr{F})$.
The Csisz\'{a}r-Kemperman-Kullback-Pinsker inequality (a.k.a. Pinsker's inequality)
states that
\begin{align}
\label{eq: Pinsker}
\tfrac12 \, |P-Q|^2 \, \log e \leq D(P \| Q)
\end{align}
where
\begin{align} \label{eq: relative entropy}
D(P \| Q) = \expectation_P\left[ \log \frac{\text{d}P}{\text{d}Q} \right]
= \int_{\set{A}} \text{d}P(a) \, \log \frac{\text{d}P}{\text{d}Q} \, (a)
\end{align}
designates the relative entropy (a.k.a. the Kullback-Leibler divergence)
from $P$ to $Q$, and
\begin{align} \label{eq: TV distance}
|P-Q| = 2 \, \sup_{\mathcal{F} \in \mathscr{F}} |P(\mathcal{F}) - Q(\mathcal{F})|
\end{align}
is the total variation distance between $P$ and $Q$.

A ``reverse Pinsker inequality" providing an upper bound on the relative
entropy in terms of the total variation distance does not exist in general
since we can find distributions that are arbitrarily close in total variation
but with arbitrarily high relative entropy.
Nevertheless, it is possible to introduce constraints under which such
reverse Pinsker inequalities can be obtained.
In the case where the probability measures $P$ and $Q$ are defined
on a common discrete (i.e., finite or countable) set~$\set{A}$,
\begin{align}
\label{eq: relative entropy - discrete}
& D(P \| Q) = \sum_{a \in \set{A}} P(a) \, \log \frac{P(a)}{Q(a)}, \\
\label{eq: total variation - discrete}
& |P-Q| = \sum_{a \in \set{A}} \bigl|P(a) - Q(a)\bigr|.
\end{align}

One of the implications of \eqref{eq: Pinsker} is that convergence
in relative entropy implies convergence in total variation distance. The total
variation distance is bounded $|P-Q| \leq 2$, whereas the relative
entropy is an unbounded information measure.

Improved versions of Pinsker's inequality were studied, e.g., in \cite{FedotovHT_IT03},
\cite{Gilardoni10}, \cite{OrdentlichW_IT2005}, \cite{ReidW11}, \cite{Vajda_IT1970}.

A ``reverse Pinsker inequality" providing an upper bound on the relative
entropy in terms of the total variation distance does not exist in general
since we can find distributions that are arbitrarily close in total variation
but with arbitrarily high relative entropy.
Nevertheless, it is possible to introduce constraints under which such
reverse Pinsker inequalities can be obtained.
In the case of a finite alphabet $\set{A}$, Csisz\'ar and Talata
\cite[p.~1012]{CsiszarT_IT06} show that
\begin{align}  \label{eq: CsTa}
D(P \| Q) \leq \left(\frac{\log e}{Q_{\min}} \right)
\cdot |P-Q|^2,
\end{align}
where
\begin{align} \label{eq: Q_min}
Q_{\min} \triangleq \min_{a \in \set{A}} Q(a).
\end{align}

Recent applications of \eqref{eq: CsTa} can be found in \cite[Appendix~D]{KostinaV15} and
\cite[Lemma~7]{TomamichelT_IT13} for the analysis of the third-order asymptotics
of the discrete memoryless channel with or without cost constraints.

In addition to $Q_{\min}$ in \eqref{eq: Q_min}, the bounds in this paper involve
\begin{align}
\label{eq: beta1}
& \beta_1 = \min_{a\in \set{A}} \frac{Q(a)}{P(a)}, \\
\label{eq: beta2}
& \beta_2 = \min_{a\in \set{A}} \frac{P(a)}{Q(a)}
\end{align}
so, $\beta_1, \beta_2 \in [0,1]$.

In this paper, Section~\ref{section: Reksnip} derives a reverse Pinsker inequality
for probability measures defined on a common finite set, improving the bound in \eqref{eq: CsTa}.
The utility of this inequality is studied in Section~\ref{sec: applications}, and it
is extended in Section~\ref{section: RD} to R\'{e}nyi divergences
of an arbitrary non-negative order.

\section{A New Reverse Pinsker Inequality for Distributions on a Finite Set}
\label{section: Reksnip}
The present section introduces a strengthened version of \eqref{eq: CsTa},
followed by some remarks and an example.

\subsection{Main Result and Proof}
\begin{theorem} \label{thm: reksnip - improved}
{\em Let $P$ and $Q$ be probability measures defined on a common finite set $\set{A}$, and assume
that $Q$ is strictly positive on $\set{A}$. Then, the following inequality holds:
\begin{align}
\label{eq: UB-RE-FS1}
D(P \| Q) & \leq \log \left(1 + \frac{|P-Q|^2}{2 Q_{\min}} \right) - \frac{\beta_2 \log e}{2} \cdot |P-Q|^2 \\
\label{eq: UB-RE-FS2}
& \leq \log \left(1 + \frac{|P-Q|^2}{2 Q_{\min}} \right)
\end{align}
where $Q_{\min}$ and $\beta_2$ are given in \eqref{eq: Q_min} and
\eqref{eq: beta2}, respectively.}
\end{theorem}

\begin{proof}
Theorem~\ref{thm: reksnip - improved} is proved by obtaining upper and lower bounds on
the $\chi^2$-divergence from $P$ to~$Q$
\begin{align} \label{eq: chi-square}
\chi^2(P\|Q) \triangleq \sum_{a \in \set{A}} \frac{(P(a)-Q(a))^2}{Q(a)}.
\end{align}
A lower bound follows by invoking Jensen's inequality
\begin{align}
\label{eq1: lb chi-square divergence}
\chi^2(P\|Q) & = \sum_{a \in \set{A}} \frac{P(a)^2}{Q(a)} - 1 \\
\label{eq2: lb chi-square divergence}
& = \sum_{a \in \set{A}} P(a) \, \exp\left( \log \frac{P(a)}{Q(a)} \right) - 1 \\
\label{eq3: lb chi-square divergence}
& \geq \exp\left(\sum_{a \in \set{A}} P(a) \, \log \frac{P(a)}{Q(a)} \right) - 1 \\[0.1cm]
\label{eq4: lb chi-square divergence}
& = \exp\bigl( D(P \| Q) \bigr) - 1.
\end{align}
Alternatively, \eqref{eq4: lb chi-square divergence} can be obtained by combining the equality
\begin{align}
\chi^2(P\|Q) = \exp\bigl(D_2(P\|Q)\bigr) - 1
\end{align}
with the monotonicity of the R\'{e}nyi divergence $D_{\alpha}(P\|Q)$ in $\alpha$, which implies
that $D_2(P\|Q) \geq D(P\|Q)$.

A refined version of \eqref{eq4: lb chi-square divergence} is derived in the following. The starting
point is a refined version of Jensen's inequality in \cite[Lemma~1]{ISSV15}, generalizing a result from
\cite[Theorem~1]{Dragomir06}), which leads to (see \cite[Theorem~7]{ISSV15})
\begin{align}
& \min_{a \in \set{A}} \frac{P(a)}{Q(a)} \cdot D(Q\|P) \nonumber \\[0.1cm]
\label{eq0: re and chi^2}
& \leq \log \bigl( 1 + \chi^2(P\|Q) \bigr) - D(P\|Q) \\[0.1cm]
\label{eq: re and chi^2}
& \leq \max_{a \in \set{A}} \frac{P(a)}{Q(a)} \cdot D(Q\|P).
\end{align}
From \eqref{eq: re and chi^2}
and the definition of $\beta_2$ in \eqref{eq: beta2}, we have
\begin{align}
& \chi^2(P\|Q) \nonumber \\[0.1cm]
\label{eq1: refined lb chi^2 divergence}
& \geq \exp \Bigl( D(P\|Q) + \beta_2 \, D(Q \|P) \Bigr) - 1 \\[0.1cm]
\label{eq2: refined lb chi^2 divergence}
& \geq \exp \left( D(P \|Q) + \frac{\beta_2 \, \log e}{2} \cdot |P-Q|^2 \right) - 1
\end{align}
where \eqref{eq1: refined lb chi^2 divergence} follows from
\eqref{eq0: re and chi^2} and the definition of $\beta_2$ in \eqref{eq: beta2},
and \eqref{eq2: refined lb chi^2 divergence} follows from Pinsker's inequality
\eqref{eq: Pinsker}. Note that the lower bound in
\eqref{eq2: refined lb chi^2 divergence} refines the lower bound in
\eqref{eq4: lb chi-square divergence} since $\beta_2 \in [0,1]$.

An upper bound on $\chi^2(P\|Q)$ is derived as follows:
\begin{align}
\chi^2(P\|Q) & = \sum_{a \in \set{A}} \frac{(P(a)-Q(a))^2}{Q(a)} \nonumber \\[0.1cm]
\label{eq1: ub1}
& \leq \frac{\sum_{a \in \set{A}} \bigl(P(a)-Q(a)\bigr)^2}{Q_{\min}} \\[0.1cm]
\label{eq2: ub1}
& = \frac{|P-Q|}{Q_{\min}} \cdot \max_{a \in \set{A}} |P(a)-Q(a)|
\end{align}
and, from \eqref{eq: TV distance},
\begin{align}
|P-Q| \geq 2 \max_{a \in \set{A}} |P(a)-Q(a)|.
\label{eq: ub2}
\end{align}
Combining \eqref{eq2: ub1} and \eqref{eq: ub2} yields
\begin{align}
\chi^2(P\|Q) \leq \frac{|P-Q|^2}{2Q_{\min}}.
\label{eq: ub chi-square divergence}
\end{align}
Finally, \eqref{eq: UB-RE-FS1} follows by combining the upper and lower bounds
on the $\chi^2$-divergence in \eqref{eq2: refined lb chi^2 divergence} and
\eqref{eq: ub chi-square divergence}.
\end{proof}

\begin{remark}
{\em It is easy to check that Theorem~\ref{thm: reksnip - improved} strengthens the
bound by Csisz\'ar and Talata in \eqref{eq: CsTa} by at least a factor of~2
since upper bounding the logarithm in \eqref{eq: UB-RE-FS1} gives
\begin{align} \label{betterthanCT}
D(P \| Q) \leq \frac{(1-\beta_2 \, Q_{\min}) \log e}{2 Q_{\min}} \cdot {|P-Q|^2}.
\end{align}}
\end{remark}

\vspace*{0.2cm}
In the finite-alphabet case, we can obtain another upper bound on $D(P\|Q)$ as a function of the $\ell_2$ norm
$\| P - Q \|_2$:
\begin{align}  \label{eq: UB-RE-FS}
D(P \| Q) \leq \log \left(1 + \frac{\|P-Q\|_2^2}{Q_{\min}} \right) - \frac{\beta_2 \log e}{2} \cdot \|P-Q\|_2^2
\end{align}
which follows by combining \eqref{eq2: refined lb chi^2 divergence}, \eqref{eq1: ub1}, and
$\|P-Q\|_2 \leq |P-Q|$. Using the inequality $\log (1 + x) \leq x \log e$ for $x \geq 0$
in the right side of \eqref{eq: UB-RE-FS}, and also loosening this bound by ignoring the
term $\frac{\beta_2 \log e}{2} \cdot \|P-Q\|_2^2$, we recover the bound
\begin{align}  \label{eq: UB-RE-FS-looser}
D(P \| Q) \leq \frac{\|P-Q\|_2^2 \, \log e}{Q_{\min}}
\end{align}
which appears in the proof of Property~4 of \cite[Lemma~7]{TomamichelT_IT13},
and also used in \cite[(174)]{KostinaV15}.

\vspace*{0.1cm}
\begin{remark}
{\em The lower bounds on the $\chi^2$-divergence in \eqref{eq4: lb chi-square divergence}
and \eqref{eq2: refined lb chi^2 divergence} improve the one in \cite[Lemma~6.3]{CsiszarT_IT06}
which states that $D(P \| Q) \leq \chi^2(P\|Q) \log e$.}
\end{remark}

\vspace*{0.1cm}
\begin{remark}
{\em Reverse Pinsker inequalities have been also derived in quantum information theory
(\cite{AE1, AE2}), providing upper bounds on the relative entropy of two quantum
states as a function of the trace norm distance when the minimal eigenvalues of the
states are positive (c.f. \cite[Theorem~6]{AE1} and \cite[Theorem~1]{AE2}).
These type of bounds are akin to the weakend form in \eqref{eq: UB-RE-FS2}. When the variational
distance is much smaller than the minimal eigenvalue (see \cite[Eq.~(57)]{AE1}),
the latter bounds have a quadratic scaling in this distance, similarly to
\eqref{eq: UB-RE-FS2}; they are also inversely proportional to the minimal eigenvalue,
similarly to the dependence of \eqref{eq: UB-RE-FS2} in $Q_{\min}$.}
\end{remark}

\section{Applications of Theorem~\ref{thm: reksnip - improved}}
\label{sec: applications}

\subsection{The Exponential Decay of the Probability for a Non-Typical Sequence}
\label{subsection: typicality}

To exemplify the utility of Theorem~\ref{thm: reksnip - improved}, we bound the function
\begin{align}\label{def:ldeltaq}
L_{\delta} ( Q ) = \min_{P \not \in \mathcal{T}_\delta (Q )} D(P\|Q)
\end{align}
where we have denoted the subset of probability measures on $ (\set{A}, \mathscr{F})$
which are $\delta$-close to $Q$ as
\begin{align}
\mathcal{T}_\delta (Q ) = \Bigl\{P
\colon \forall \, a \in \set{A}, \; \; |P(a)-Q(a)| \leq \delta \, Q(a) \Bigr\}
\end{align}
Note that $(a_1, \ldots , a_n )$ is strongly $\delta$-typical according to $Q$ if its
empirical distribution belongs to $\mathcal{T}_\delta (Q )$. According to Sanov's theorem
(e.g. \cite[Theorem~11.4.1]{Cover_Thomas}), if the random variables are independent
distributed according to $Q$, then the probability that
$(Y_1, \ldots , Y_n)$,  is not $\delta$-typical
vanishes exponentially with exponent $L_{\delta} ( Q )$.
\par
To state the next result, we invoke the following notions from \cite{OrdentlichW_IT2005}.
Given a probability measure $Q$, its \textit{balance coefficient} is given by
\begin{align} \label{balanceOW}
\beta_Q = \inf_{A\in \mathscr{F}\colon Q(A) \geq \frac12} Q(A).
\end{align}
The function $\phi\colon(0, \tfrac12] \to [\tfrac12 \log e, \infty)$ is given by
\begin{align}
\label{eq: phi refined pinsker}
\phi(p) = \left\{
\begin{array}{ll}
\frac1{4(1-2p)} \, \log \left( \frac{1-p}{p} \right), & p \in
\bigl(0, \tfrac12 \bigr), \\[0.2cm]
\tfrac12 \log e, & p=\tfrac12 .
\end{array}
\right.
\end{align}
\begin{theorem}
{\em If $Q_{\min}>0$, then
\begin{align}
\label{potalo}
\phi(1 - \beta_Q) \, Q_{\min}^2 \, \delta^2
&\leq L_{\delta} ( Q ) \\
&\leq \log \left(1 + 2 Q_{\min} \, \delta^2 \right)
\label{potaup}
\end{align}
where \eqref{potaup} holds if $\delta \leq Q_{\min}^{-1} -1$.}
\end{theorem}

\vspace*{0.2cm}

\begin{proof}
Ordentlich and Weinberger \cite[Section~4]{OrdentlichW_IT2005}
show the refinement of Pinsker's inequality:
\begin{align}  \label{eq: OrdentlichW}
\phi(1-\beta_Q) \; |P-Q|^2 \leq D(P \| Q).
\end{align}
Note that if $Q_{\min} > 0$ then $\beta_Q \leq 1 - Q_{\min} < 1$, and
therefore $\phi(1-\beta_Q)$ is well defined and finite.
If $P \not\in \mathcal{T}_\delta (Q )$ the simple bound
\begin{align}
|P - Q| > \delta Q_{\min}
\end{align}
together with \eqref{eq: OrdentlichW} yields \eqref{potalo}.
\par
The upper bound \eqref{potaup} follows from \eqref{eq: UB-RE-FS2}
and the fact that if $\delta \leq Q_{\min}^{-1}-1$, then
\begin{align}
\min_{P \not\in \mathcal{T}_\delta (Q )} | P - Q | = 2 \delta Q_{\min}.
\end{align}
\end{proof}

If $\delta \leq Q_{\min}^{-1} - 1$, the ratio between the upper and lower bounds
in \eqref{potaup}, satisfies
\begin{align}  \label{eq: exponents' ratio}
 \frac1{Q_{\min}} \cdot \frac{\log e}{2 \, \phi(1-\beta_Q)} \cdot
\frac{\log \left(1 + 2 Q_{\min} \, \delta^2 \right)}{\tfrac12 \log e
\; Q_{\min} \, \delta^2} \leq \frac{4}{Q_{\min}}
\end{align}
where \eqref{eq: exponents' ratio} follows from the fact that its second
and third factors are less than or equal to~1 and~4, respectively. Note
that the bounds in \eqref{potalo} and \eqref{potaup} scale like $\delta^2$
for $\delta \approx 0$.

\subsection{Distance from Equiprobable}
\label{subs: equiprobable}

If $P$ is a distribution on a finite set $\set{A}$, $H ( P )$ gauges the ``distance" from
$U$, the equiprobable distribution, since
\begin{align} \label{eq: RE-equiprobable}
H ( P ) = \log | \set{A} | -  D(P \| U).
\end{align}
Thus, it is of interest to explore the relationship between $H ( P ) $ and $|P-U|$.
Particularizing \eqref{eq: Pinsker},  \cite[(2.2)]{BretagnolleH79} (see also \cite[pp.~30--31]{Vapnik98}),
and \eqref{eq: UB-RE-FS2} we obtain
\begin{align}
\label{eq: 1st ubtv-uniform}
|P-U| &\leq \sqrt{ \frac{2}{\log e} \cdot \bigl( \log | \set{A} | - H(P) \bigr) }, \\
\label{eq: 2nd ubtv-uniform}
|P-U| &\leq 2 \sqrt{1 - \frac1{|\set{A}|} \cdot \exp\bigl( H(P) \bigr)}, \\
\label{eq: lbtv-uniform}
|P-U| &\geq \sqrt{ 2 \left( \exp\bigl( -H(P) \bigr) - \frac1{|\set{A}|} \right) },
\end{align}
respectively.

\begin{figure}[here!]
\includegraphics[width=8.2cm]{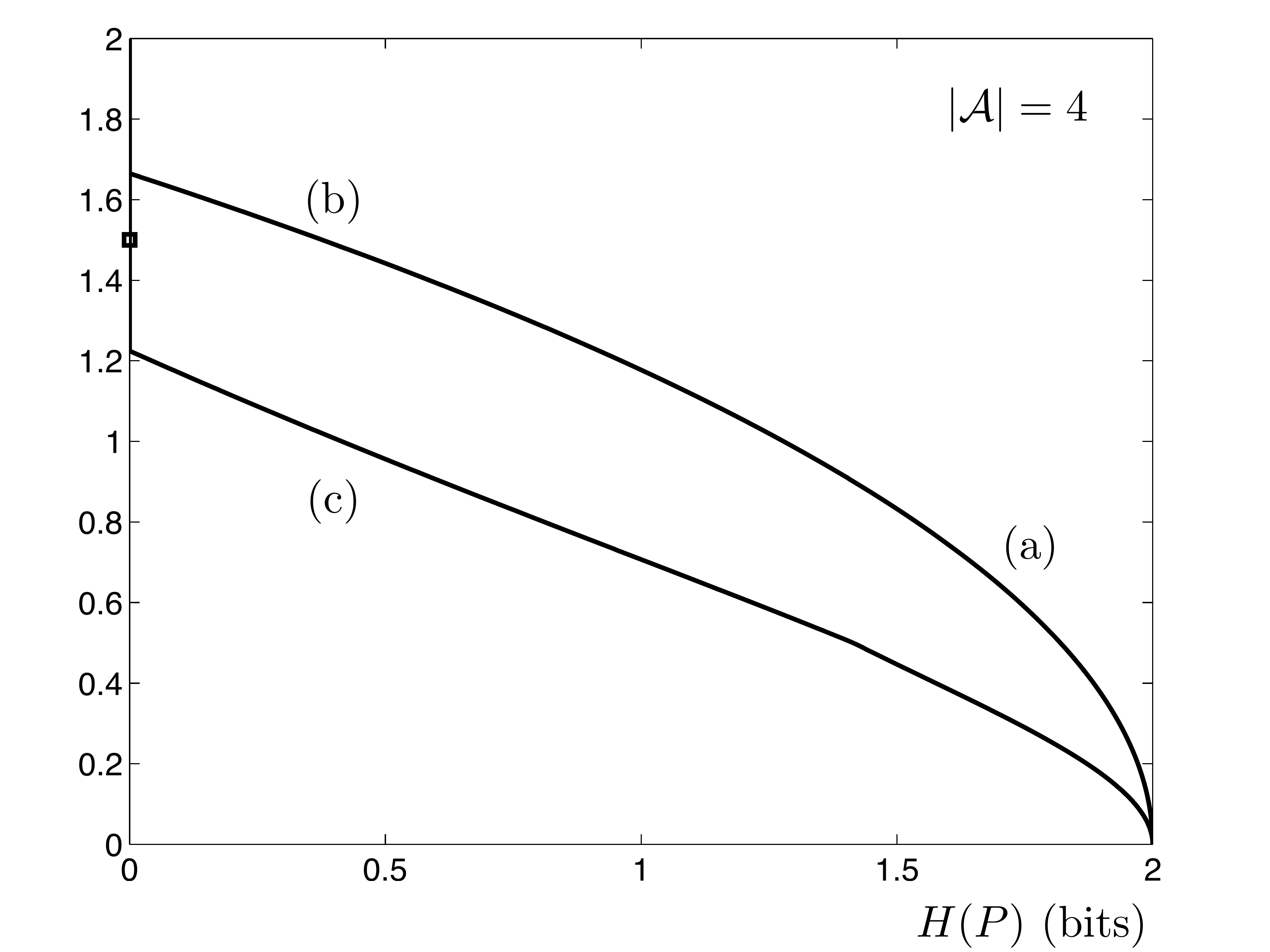}
\includegraphics[width=8.2cm]{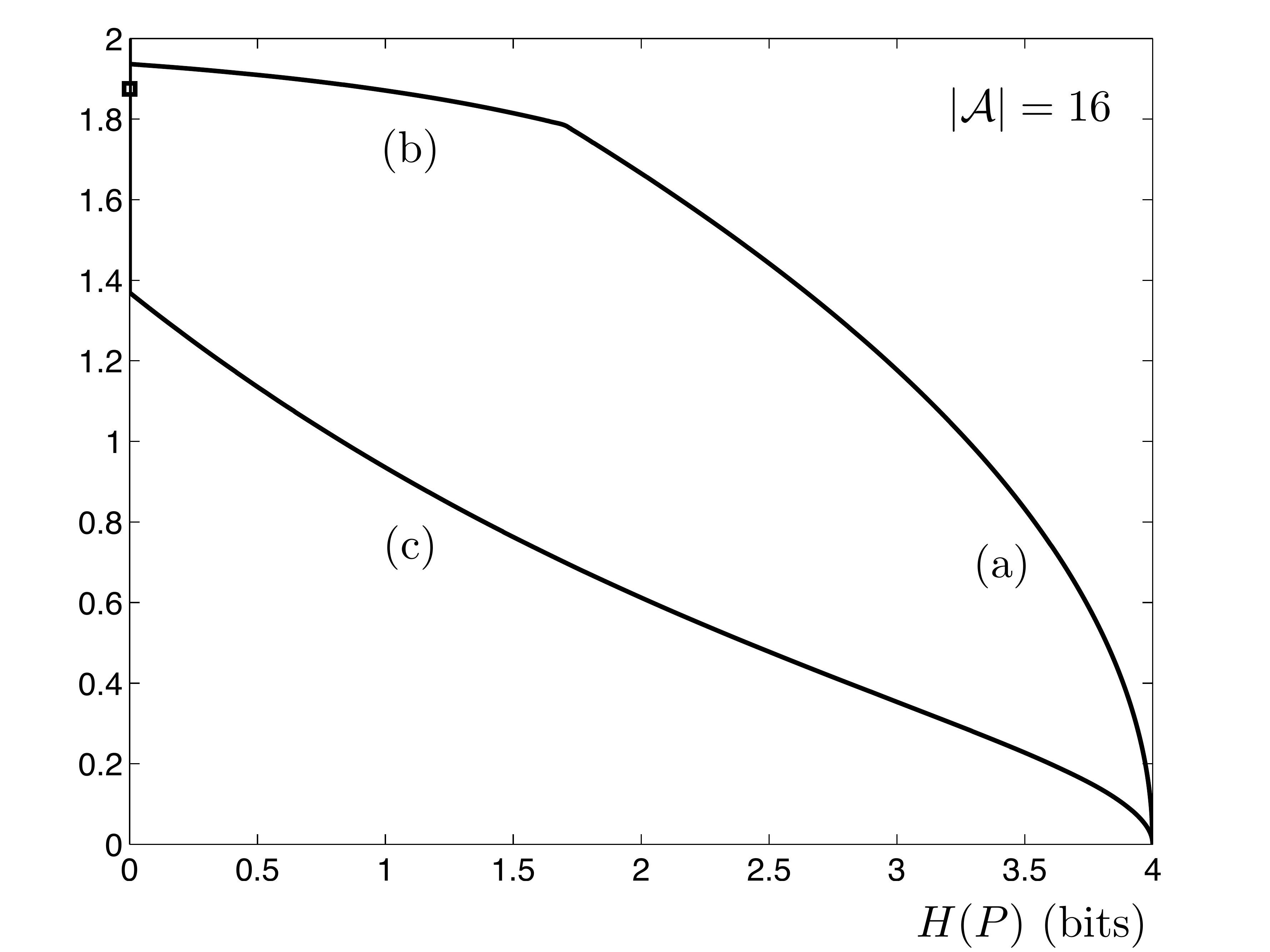}
\caption{\label{figure:equiprobable}
Bounds on $|P-U|$ as a function of $H(P)$ for $|\set{A}| = 4$, and
$|\set{A}| = 16$. The point $(H(P), |P-U|)=(0, 2(1-|\set{A}|^{-1}) )$
is depicted on the $y$-axis. In the curves of the two plots, the bounds
(a), (b) and (c) refer, respectively, to \eqref{eq: 1st ubtv-uniform},
\eqref{eq: 2nd ubtv-uniform} and \eqref{eq: lbtv-uniform}.}
\end{figure}

The bounds in \eqref{eq: 1st ubtv-uniform}--\eqref{eq: lbtv-uniform} are illustrated
for $|\set{A}| = 4, 16$ in Figure~\ref{figure:equiprobable}.
For $H(P)=0$, $|P - U| = 2(1-|\set{A}|^{-1} )$ is shown for reference in Figure~\ref{figure:equiprobable};
as the cardinality of the alphabet increases,
the gap between $|P-U|$ and its upper bound is reduced (and this gap decays
asymptotically to zero).
\par
Results on the more general problem of finding bounds on $|H( P ) - H(Q)|$ based on $|P-Q|$
can be found in \cite[Theorem~17.3.3]{Cover_Thomas}, \cite{HoY_IT2010},
\cite{Prelov_PPI2008}, \cite{Sason_IT2013}, \cite[Section~1.7]{Verdu_book} and \cite{Zhang_IT2007}.

\section{Extension of Theorem~\ref{thm: reksnip - improved} to R\'{e}nyi Divergences}
\label{section: RD}

\begin{definition}
{\em The R\'{e}nyi divergence of order $\alpha \in [0, \infty]$ from $P$ to $Q$ is defined
for $\alpha \in (0,1) \cup (1, \infty)$ as
\begin{align}
& D_{\alpha}(P || Q) \triangleq \frac{1}{\alpha-1} \; \log \left( \sum_{a \in \mathcal{A}}
P^{\alpha}(a) \, Q^{1-\alpha}(a) \right).
\label{eq: Renyi divergence}
\end{align}
Recall that $D_{1}(P\|Q) \triangleq D(P\|Q)$ is defined to be the analytic extension
of $D_{\alpha}(P || Q)$ at $\alpha=1$ (if $D(P||Q) < \infty$, L'H\^{o}pital's rule
gives that $D(P||Q) = \lim_{\alpha \uparrow 1} D_{\alpha}(P || Q)$).
The extreme cases of $\alpha = 0, \infty$ are defined as follows:
\begin{itemize}
\item If $\alpha = 0$ then $D_0(P||Q) = -\log Q(\text{Support}(P))$,
\item If $\alpha = +\infty$ then
$$D_{\infty}(P||Q) = \log \left(\sup_{a \in \set{A}} \frac{P(a)}{Q(a)}\right).$$
\end{itemize}}
\end{definition}

Pinsker's inequality was extended by Gilardoni \cite{Gilardoni10}
for a R\'{e}nyi divergence of order $\alpha \in (0,1]$
(see also \cite[Theorem~30]{ErvenH14}), and it gets the form
\begin{equation*}
\tfrac{\alpha}{2} \, |P-Q|^2 \, \log e \leq D_{\alpha}(P \| Q).
\end{equation*}
A tight lower bound on the R\'{e}nyi divergence of order
$\alpha > 0$ as a function of the total variation distance
is given in \cite{Sason_ISIT15}, which is consistent with
Vajda's tight lower bound for $f$-divergences in \cite[Theorem~3]{Vajda_1972}.

Motivated by these findings, we extend the upper bound on the relative entropy
in Theorem~\ref{thm: reksnip - improved} to R\'{e}nyi divergences of an arbitrary order.

\vspace*{0.2cm}
\begin{theorem}\label{thm: ub-RD-TV-FS}
{\em
Assume that $P, Q$ are strictly positive with minimum masses
denoted by $P_{\min}$ and $Q_{\min}$, respectively.
Let $\beta_1$ and $\beta_2$ be given in \eqref{eq: beta1} and \eqref{eq: beta2},
respectively, and abbreviate $\delta \triangleq \tfrac12 |P-Q| \in [0,1]$.
Then, the R\'{e}nyi divergence of order $\alpha \in [0,\infty]$ satisfies
\begin{align}  \label{eq: ub-RD-TV-FS}
& D_{\alpha}(P \| Q) \nonumber \\
& \leq \left\{
\begin{array}{ll}
f_1, & \mbox{$\alpha \in (2, \infty]$} \\[0.3cm]
f_2, & \mbox{$\alpha \in [1,2]$} \\[0.3cm]
\min \left\{ f_2, f_3, f_4 \right\},
& \mbox{$\alpha \in \bigl(\tfrac12, 1 \bigr)$} \\[0.3cm]
\min \left\{ 2 \log\left(\frac1{1-\delta} \right), f_2, f_3, f_4 \right\},
& \mbox{$\alpha \in \bigl[0, \tfrac12\bigr]$}
\end{array}
\right.
\end{align}
where, for $\alpha \in [0, \infty]$,
\begin{align} \label{eq: f1}
& f_1(\alpha, \beta_1, \delta) \nonumber \\
& \triangleq \left\{
\begin{array}{ll}
\frac{1}{\alpha-1} \; \log\left(1 + \frac{\delta (\beta_1^{1-\alpha}-1)}{1-\beta_1}\right)
& \alpha \in [0,1) \cup (1, \infty) \\[0.3cm]
\frac{\delta}{1-\beta_1} \; \log \frac1{\beta_1}, & \alpha = 1, \\[0.3cm]
\log \frac1{\beta_1}, & \alpha = \infty \\[0.3cm]
\end{array}
\right.
\end{align}
for $\alpha \in [0,2]$
\begin{align}
\label{eq: f2}
& f_2(\alpha, \beta_1, Q_{\min}, \delta) \nonumber \\[0.2cm]
& \triangleq \min \left\{f_1(\alpha, \beta_1, \delta), \; \log\left(1+\frac{2\delta^2}{Q_{\min}}\right) \right\}
\end{align}
and, for $\alpha \in [0,1)$, $f_3$ and $f_4$ are given by
\begin{align}
\label{eq: f3}
&  f_3(\alpha, P_{\min}, \beta_1, \delta) \nonumber \\[0.2cm]
& \triangleq
\left(\frac{\alpha}{1-\alpha}\right)
\left[ \log \left(1 + \frac{2 \delta^2}{P_{\min}} \right)
- 2 \beta_1 \delta^2 \, \log e  \right], \\[0.3cm]
\label{eq: f4}
&  f_4(\beta_2, Q_{\min}, \delta) \nonumber \\
& \triangleq \min\left\{ \log \left(1 + \frac{2 \delta^2}{Q_{\min}} \right)
- 2 \beta_2 \delta^2 \, \log e , \right.
\nonumber \\[0.1cm]
& \hspace*{1.2cm} \left. \log \left(1 + \frac{\min\{\delta, 2\delta^2\}}{Q_{\min}} \right) \right\}.
\end{align}}
\end{theorem}

\vspace*{0.3cm}
\begin{proof}
See \cite[Section~7.C]{ISSV15}.
\end{proof}

\begin{remark}
{\em A simple bound, albeit looser than the one in Theorem~\ref{thm: ub-RD-TV-FS} is
\begin{align} \label{eq: simple ub}
D_{\alpha}(P\|Q) \leq \log \left(1 + \frac{|P-Q|}{2 Q_{\min}} \right)
\end{align}
which is asymptotically tight as $\alpha \to \infty$ in the case of a binary alphabet
with equiprobable $Q$.}
\end{remark}

\begin{example}
{\em Figure~\ref{figure:compare_ub_RD_finite_graph1} illustrates the bound in
\eqref{eq: f1}, which is valid for all $\alpha \in [0, \infty]$ (see
\cite[Theorem~23]{ISSV15}), and the upper bounds of Theorem~\ref{thm: ub-RD-TV-FS} 
in the case of binary alphabets.
\begin{figure}[here!]
\centerline{\includegraphics[width=9.8cm]{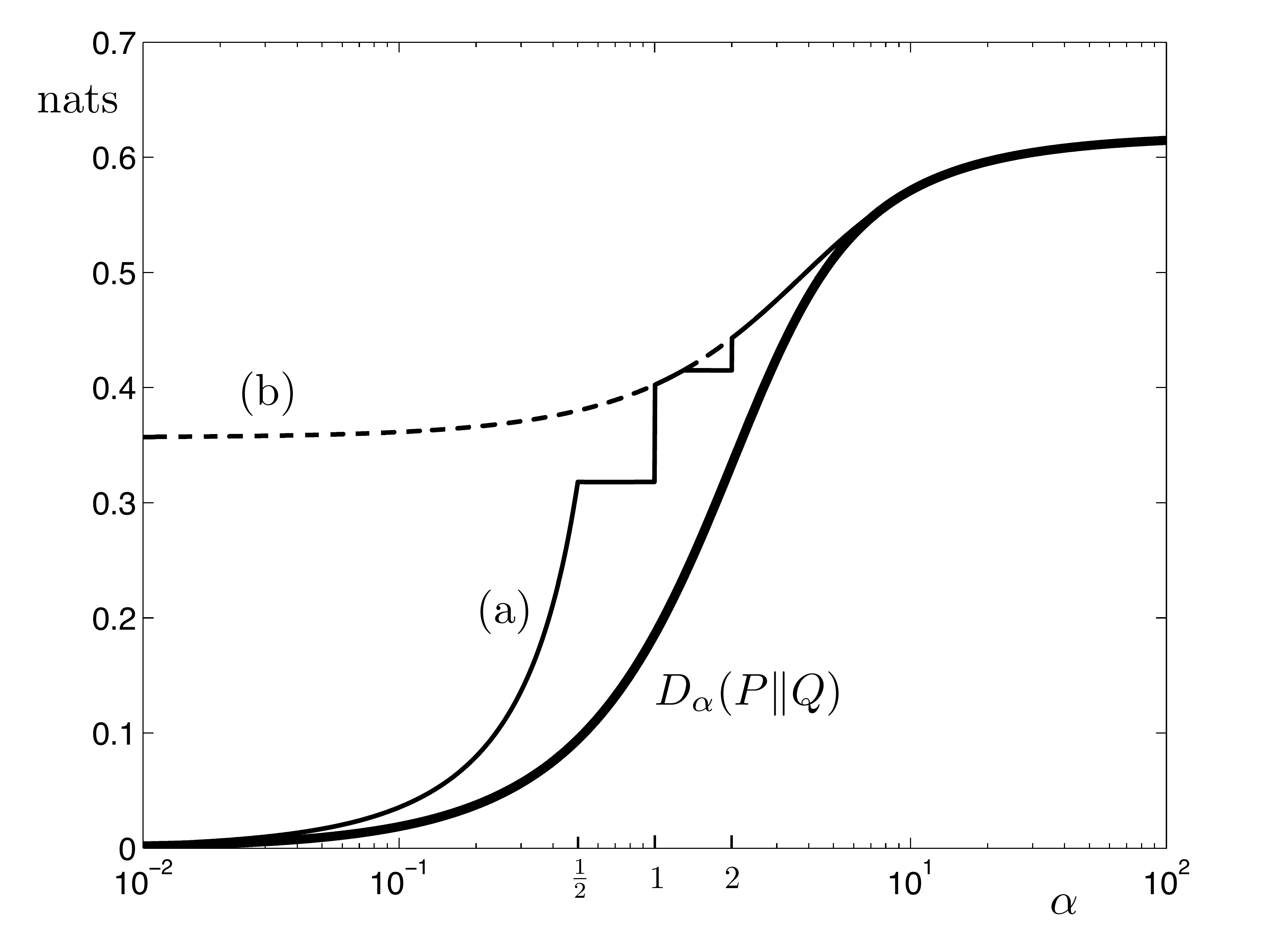}}
\caption{\label{figure:compare_ub_RD_finite_graph1}
The R\'{e}nyi divergence $D_{\alpha}(P\|Q)$ for $P$ and $Q$
which are defined on a binary alphabet with $P(0)=Q(1)=0.65$,
compared to (a) its upper bound in \eqref{eq: ub-RD-TV-FS},
and (b) its upper bound in \eqref{eq: f1} (see \cite[Theorem~23]{ISSV15}).
The two bounds coincide here when $\alpha \in (1, 1.291) \cup (2, \infty)$.}
\end{figure}}
\end{example}

\section{Summary}
We derive in this paper some ``reverse Pinsker inequalities" for probability measures
$P \ll Q$ defined on a common finite set, which provide lower bounds on the total
variation distance $P-Q$ as a function of the relative entropy $D(P\|Q)$ under the
assumption of a bounded relative information or $Q_{\min} > 0$. More general results
for an arbitrary alphabet are available in \cite[Section~5]{ISSV15}.

In \cite{ISSV15}, we study bounds among various $f$-divergences, dealing with arbitrary
alphabets and deriving bounds on the ratios of various distance measures.
New expressions of the R\'{e}nyi divergence in terms of the relative information
spectrum are derived, leading to upper and lower bounds on the R\'{e}nyi divergence
in terms of the variational distance.

\section*{Acknowledgment}
The work of I. Sason has been supported by the Israeli Science Foundation (ISF) under
Grant 12/12, and the work of S. Verd\'{u} has been supported by the US National Science
Foundation under Grant CCF-1016625, and in part by the Center for Science of Information,
an NSF Science and Technology Center under Grant CCF-0939370.

\end{document}